\documentclass[11pt]{article}
\usepackage{latexsym}
\usepackage{amsmath}
\usepackage{amssymb}
\usepackage{amsfonts,amsthm}
\usepackage[top=1in, bottom=1in, left=1in, right=1in]{geometry}
\newcommand{\comment}[1]{}

\usepackage{tikz}
\usetikzlibrary{arrows,decorations.pathmorphing,decorations.shapes,backgrounds,fit}
\pgfkeys{tikz/cc1/.style={dotted} }
\pgfkeys{tikz/cc4/.style={dashed} }
\pgfkeys{tikz/cc3/.style={decorate,decoration=bumps} }
\pgfkeys{tikz/cc2/.style={decorate,decoration=snake} }
\pgfkeys{tikz/solid/.style={line width=2pt} }
\pgfkeys{tikz/solid2/.style={line width=1pt} }
\pgfkeys{tikz/nodostile1/.style={draw,line width=1pt}}
\pgfkeys{tikz/nodostile2/.style={draw,line width=2pt}}
\pgfkeys{tikz/arcostile1/.style={very thick} }
\pgfkeys{tikz/arcostile2/.style={red,very thick} }
\pgfkeys{tikz/arcostile3/.style={thick} }
\pgfkeys{tikz/arcostile4/.style={line width=2pt} }

\newtheorem{theorem}{Theorem}[section]
\newtheorem{lemma}[theorem]{Lemma}

\newtheorem{proposition}[theorem]{Proposition}

\newcommand{\tabhead}{\sc }

\newcommand{\algo}[1]{{\textsc{#1}}}

\newcommand{\ignore}[1]{}

\makeatletter
\def\argmin{\mathop{\operator@font argmin}}
\makeatother

\begin{document}

\title{An Experimental Study of Dynamic Dominators\thanks{A preliminary version of this paper appeared in the \emph{Proceedings of the 20th Annual European Symposium on Algorithms}, pages 491--502, 2012.}}

\author{Loukas Georgiadis$^{1}$
\and Giuseppe F. Italiano$^{2}$
\and Luigi Laura$^{3}$
\and Federico Santaroni$^{2}$
}

\footnotetext[1]{Department of Computer Science \& Engineering, University of Ioannina, Greece. E-mail: \texttt{loukas@cs.uoi.gr}.}
\footnotetext[2]{Dipartimento di Ingegneria Civile e Ingegneria Informatica, Universit\`a di Roma ``Tor Vergata'', Roma, Italy. E-mail: \texttt{giuseppe.italiano@uniroma2.it, santaroni@uniroma2.it}.}
\footnotetext[3]{Dipartimento di Ingegneria Informatica, Automatica e Gestionale e Centro di Ricerca per il Trasporto e la Logistica (CTL), ``Sapienza'' Universit\`a di Roma, Roma, Italy. E-mail: \texttt{laura@dis.uniroma1.it}.}

\maketitle

\begin{abstract}
Motivated by recent applications of dominator computations, we consider the problem of dynamically maintaining the dominators of  flow graphs through a sequence of insertions and deletions of edges.
Our main theoretical contribution is a simple incremental algorithm that maintains the dominator tree of a flow graph with $n$ vertices through a sequence of $k$ edge insertions in $O(m\min\{n,k\}+kn)$ time, where $m$ is the total number of edges after all insertions.
Moreover, we can test in constant time if a vertex $u$ dominates a vertex $v$, for any pair of query vertices $u$ and $v$.
Next, we present a new decremental algorithm to update a dominator tree through a sequence of edge deletions. Although our new decremental algorithm is not asymptotically faster than repeated applications of a static algorithm, i.e., it runs in $O(mk)$ time for $k$ edge deletions, it performs well in practice. By combining our new incremental and decremental algorithms we obtain a fully dynamic algorithm that maintains the dominator tree through intermixed sequence of insertions and deletions of edges.
Finally, we present efficient implementations of our new algorithms as well as of existing algorithms, and conduct an extensive experimental study on real-world graphs taken from a variety of application areas.
\end{abstract}

\section{Introduction}
\label{sec:intro}

A flow graph $G=(V,E,s)$ is a directed graph with a distinguished start vertex $s \in V$. A vertex $v$ is \emph{reachable} in $G$ if there is a path from $s$ to $v$; $v$ is \emph{unreachable} if no such path exists. The \emph{dominator relation} on $G$ is defined for the set of reachable vertices as follows.
A vertex $w$ \emph{dominates} a vertex $v$ if every path from $s$ to $v$ includes $w$.
We let $\mathit{Dom}(v)$ denote the set of all vertices that dominate $v$. If $v$ is reachable then $\mathit{Dom}(v) \supseteq \{s ,v\}$; otherwise $\mathit{Dom}(v) = \emptyset$.
For a reachable vertex $v$, $s$ and $v$ are its \emph{trivial dominators}. A vertex $w \in \mathit{Dom}(v)-v$ is a \emph{proper dominator} of $v$.
The \emph{immediate dominator} of a vertex $v\neq s$, denoted $d(v)$, is the unique vertex $w\neq v$ that dominates $v$ and is dominated by all vertices in
$\mathit{Dom}(v)-v$. The dominator relation is reflexive and transitive.
Its transitive reduction is a rooted tree, the \emph{dominator tree} $D$: $u$ dominates $w$ if and only if $u$ is an ancestor of $w$ in $D$.
To form $D$, we make each reachable vertex $v \not= s$ a child of its immediate dominator.

The problem of finding dominators has been extensively studied, as it occurs in several applications.
The dominator tree is a central tool in program optimization and code generation~\cite{cytron:91:toplas}. Dominators have also been used in constraint programming~\cite{QVDR:PADL:2006},
circuit testing~\cite{amyeen:01:vlsitest}, theoretical biology~\cite{foodwebs:ab04}, memory profiling~\cite{memory-leaks:mgr2010},
fault-tolerant computing~\cite{FaultTolerantReachability},
connectivity and path-determination problems~\cite{2vc,2VCSS:Geo,2ECB,2VCB,2CC:HenzingerKL15,Italiano2012,2vcb:jaberi15,2VCC:Jaberi2016}, and the analysis of diffusion networks~\cite{Rodrigues:icml12}.
Allen and Cocke showed that the dominator relation can be
computed iteratively from a set of data-flow equations~\cite{cf:ac}. A direct implementation of this method has an
$O(mn^2)$ worst-case time bound, for a flow graph with $n$ vertices and $m$ edges. Cooper et al.~\cite{dom:chk01} presented a clever
tree-based space-efficient implementation of the
iterative algorithm. Although it does not improve the $O(mn^2)$ worst-case time bound, the tree-based version is much more efficient in practice.
Purdom and Moore~\cite{dom:pm72}
gave an algorithm, based on reachability, with complexity $O(mn)$.
Improving on previous work by Tarjan~\cite{domin:tarjan}, Lengauer
and Tarjan~\cite{domin:lt} proposed an $O(m \log_{(m/n + 1)}{n})$-time
algorithm and a more complicated $O(m \alpha(m,n))$-time version,
where $\alpha(m,n)$ is an extremely slow-growing functional
inverse of the Ackermann function~\cite{dsu:tarjan}.  Subsequently, more-complicated but truly linear-time algorithms to compute $D$ were discovered~\cite{domin:ahlt,dominators:bgkrtw,domin:bkrw,dom:gt04}, as well as near-linear-time algorithms that use simple data structures \cite{dominators:Fraczak2013,dominators:poset}.

An experimental study of static algorithms for computing dominators was presented in
\cite{dom_exp:gtw}, where careful
implementations of both versions of the Lengauer-Tarjan algorithm, the iterative algorithm of Cooper et al., and a
new hybrid algorithm (\algo{snca}) were given. In these experimental results
the performance of all these algorithms was similar, but
the simple version of the Lengauer-Tarjan algorithm
and the hybrid algorithm were most consistently fast, and their advantage increased as the input
graph got bigger or more complicated. The graphs used in~\cite{dom_exp:gtw} were taken from the application areas mentioned above and have moderate size
(at most a few thousand vertices and edges) and simple enough structure that they can be efficiently processed by the iterative algorithm.
Recent experimental results for computing dominators in large graphs are reported in~\cite{strong-articulation:algorithmica,2VCSS:Geo,DomCertExp:SEA13,Loops:SEA14}.
There it is apparent that the simple iterative algorithms are not competitive with the more sophisticated algorithms based on Lengauer-Tarjan
for larger and more complicated graphs.
\ignore{The graphs used in these experiments were taken from applications of dominators in memory profiling~\cite{memory-leaks:mgr2010,objown:mitchell06}, testing 2-vertex connectivity and computing sparse 2-vertex connected spanning subgraphs~\cite{2vc,2VCSS:Geo}, and computing strong articulation points and strong bridges in directed graphs~\cite{Italiano2012}, which typically involve much larger and complicated graphs.}

Here we consider the problem of dynamically maintaining the dominator relation of a flow graph that undergoes both insertions and deletions of edges. Vertex insertions and deletions can be simulated using combinations of edge updates. We recall that a dynamic graph problem is said to be \emph{fully dynamic} if it requires to process both insertions and deletions of edges, \emph{incremental} if it requires to process edge insertions only and \emph{decremental} if it requires to process edge deletions only.
The fully dynamic dominators problem arises in various applications, such as data flow analysis and compilation~\cite{semidynamic-digraphs:CFNP}. Moreover,~\cite{2vc,Italiano2012} imply that a fully dynamic dominators algorithm can be used for dynamically testing 2-vertex connectivity, and maintaining the strong articulation points of a digraph. The decremental dominators problem appears in the computation of 2-connected components in digraphs~\cite{LuigiGILP15,2CC:HenzingerKL15,2VCC:Jaberi2016}.

The problem of updating the dominator relation has been studied for few decades (see, e.g.,~\cite{dynamicdominator:AL,incrementaldominators:CR,semidynamic-digraphs:CFNP,PGT11,RR:incdom,SGL}). However, a
worst-case complexity bound for a single update better than $O(m)$ has been only achieved for special cases, mainly for incremental or decremental problems.
Specifically, the algorithm of Cicerone et al.~\cite{semidynamic-digraphs:CFNP} achieves $O(n \max\{k,m_0\} + q)$ running time for processing a sequence of $k$ edge insertions interspersed with $q$ queries of the type ``does $u$ dominate $v$?'', for a flow graph with $n$ vertices and initially $m_0$ edges.
This algorithm, however, requires $O(n^2)$ space, as it needs to maintain the transitive closure of the graph.
The same bounds are also achieved for a sequence of $k$ deletions, but only for a \emph{reducible} flow graph (defined below).
Alstrup and Lauridsen describe in a technical report~\cite{dynamicdominator:AL} an algorithm that maintains the dominator tree through a sequence of $k$ edge insertions interspersed with $q$ queries in $O(m \min\{k,n\} + q)$ time. In this bound $m$ is the number of edges after all insertions.
Unfortunately, the description and the analysis of this algorithm are incomplete.
Our main theoretical contribution is to provide a simple incremental algorithm that
maintains the dominator tree through a sequence of $k$ edge insertions in $O(m \min\{k,n\} + kn)$ time.
We can also answer dominance queries (test if a vertex $u$ dominates another vertex $v$) in constant time.
Moreover, we provide an efficient implementation of this algorithm that performs very well in practice.

Although theoretically efficient solutions to the fully dynamic dominators problem appear still beyond reach, there is a need for practical algorithms and fast implementations in several application areas. In this paper, we also present new fully dynamic dominators algorithms and efficient implementations of known algorithms, such as the algorithm by Sreedhar, Gao and Lee~\cite{SGL}. We evaluate the implemented algorithms experimentally using real data taken from the application areas of dominators. To the best of our knowledge, the only previous experimental study of (fully) dynamic dominators algorithms appears in~\cite{PGT11}; here we provide new algorithms, improved implementations, and an experimental evaluation using bigger graphs taken from a larger variety of applications.
Other previous experimental results, reported in~\cite{sreedhar:phd} and the references therein, are limited to comparing incremental algorithms against the static computation of dominators.

\section{Basic definitions and properties}
\label{sec:properties}

The algorithms we consider can be stated in terms of two structural properties of dominator trees that we discuss next. Let $G=(V,E,s)$ be a flow graph, and let $T$ be a tree rooted at $s$, with vertex set consisting of the vertices that are reachable from $s$.
For a reachable vertex $v \not= s$, let $t(v)$ denote the parent of $v$ in $T$.
Tree $T$ has the \emph{parent property} if for all $(v, w) \in E$ such that $v$ is reachable,
$v$ is a descendant of $t(w)$ in $T$. If $T$ has the parent property then $t(v)$ dominates $v$ for every reachable vertex $v \not= s$~\cite{DomCert:TALG}.
The next lemma states another useful property of trees that satisfy the parent property.

\begin{lemma}\cite{DomCert:TALG}
\label{lemma:wide-path}
Let $T$ we a tree with the parent property. If $v$ is an ancestor of $w$ in $T$, there is a path from $v$ to $w$ in $G$, and every vertex on a simple path from $v$ to $w$ in $G$ is a descendant of $v$ but not a proper descendant of $w$ in $T$.
\end{lemma}

Let $T$ be a tree with the parent property, and let $v$ be a reachable vertex of $G$.
We define the \emph{support} $\mathit{sp}_T(v,w)$ of an edge $(v,w)$ with respect to $T$ as follows: if $v = t(w)$, $\mathit{sp}_T(v,w) = v$; otherwise, $\mathit{sp}_T(v,w)$ is the child of $t(w)$ that is an ancestor of $v$.

Tree $T$ has the \emph{sibling property} if $v$ does not dominate $w$ for all siblings $v$ and $w$. The parent and sibling properties are necessary and sufficient for a tree to be the dominator tree.

\begin{theorem}\cite{DomCert:TALG}
\label{thm:wide-narrow-dom}
Tree $T$ is the dominator tree ($T=D$) if and only if it has the parent and the sibling properties.
\end{theorem}

Now consider the effect that a single edge update (insertion or deletion) has on the dominator tree $D$. Let $(x,y)$ be the inserted or deleted edge. We let $G'$ and $D'$ denote the flow graph and its dominator tree after the update. Similarly, for any function $f$ on $V$, we let $f'$ be the function after the update.
By definition, $D' \not= D$ only if $x$ is reachable before the update. We say that a vertex $v$ is \emph{affected} by the update if $d'(v) \not= d(v)$. (Note that we can have $\mathit{Dom}'(v) \not= \mathit{Dom}(v)$ even if $v$ is not affected.)
If $v$ is affected then $d(v)$ does not dominate $v$ in $G'$. This implies that the insertion of $(x,y)$ creates a path from $s$ to $v$ that
avoids $d(v)$.

The difficulty in updating the dominance relation lies on two facts:
(i) An affected vertex can be arbitrarily far from the updated edge, and (ii) a single update may affect many vertices.
Two pathological examples are shown in Figures \ref{fig:pathological1} and \ref{fig:pathological2}.
The graph family of Figure \ref{fig:pathological1} contains, for any $n \ge 3$, a directed graph with $n$ vertices that consists of the path $(s=v_0, v_1, v_2, \ldots, v_{n-1})$, together with the reverse subpath $(v_{n-1}, v_{n-2}, \ldots, v_2)$. Initially we have $d(v_i)=v_{i-1}$ for all $i \in \{1, 2,\ldots, n-1\}$. The insertion of edge $(s,v_{n-1})$ makes $d'(v_i)=s$ for all $i \in \{2,3,\ldots,n-1\}$, while the deletion of $(s, v_{n-1})$ restores the initial dominator tree. So each single edge update affects every vertex except $s$ and $v_1$. A similar example, but for a family of directed acyclic graphs is shown in Figure \ref{fig:pathological2}. This graph family contains, for any $n \ge 3$, a directed acyclic graph with $n$ vertices that consists of the path $(s=v_0, v_1, v_2, \ldots, v_{\lfloor n/2 \rfloor -1})$, together with the edges $(v_{\lfloor n/2 \rfloor -1}, v_i)$ for all $i \in \{\lfloor n/2 \rfloor, \lfloor n/2 \rfloor+1, \ldots, n-1\}$. Initially we have $d(v_i)=v_{i-1}$ for all $i \in \{1, 2,\ldots, \lfloor n/2 \rfloor - 1\}$, and  $d(v_i)=v_{\lfloor n/2 \rfloor-1}$ for all $i \in \{\lfloor n/2 \rfloor, \lfloor n/2 \rfloor+1, \ldots, n-1\}$. The insertion of edge $(s,v_{n-1})$ makes $d'(v_i)=s$ for all $i \in \{\lfloor n/2 \rfloor, \lfloor n/2 \rfloor+1, \ldots, n-1\}$, while the deletion of $(s, v_{n-1})$ restores the initial dominator tree. So each single edge update affects $\lceil n/2 \rceil$ vertices.
Moreover, we can construct sequences of $\Theta(n)$ edge insertions (deletions) such that each single insertion (deletion) affects $\Theta(n)$ vertices. Consider, for instance, the graph family of Figure~\ref{fig:pathological1} and the sequence of insertions $(v_{n-3},v_{n-1}), (v_{n-4},v_{n-1}), \ldots, (s,v_{n-1})$, or the the graph family of Figure~\ref{fig:pathological2} and the sequence of insertions $(v_{\lfloor n/2 \rfloor - 2}, v_{n-1}), (v_{\lfloor n/2 \rfloor - 3}, v_{n-1}), \ldots, (s, v_{n-1})$. This implies a lower bound of $\Omega(n^2)$ time for any algorithm that maintains $D$ (or the complete dominator relation) explicitly through a sequence of $\Omega(n)$ edge insertions or a sequence of $\Omega(n)$ edge deletions, and a lower bound of $\Omega(mn)$ time for any algorithm that maintains $D$ (or the complete dominator relation) explicitly through an intermixed sequence of $\Omega(m)$ edge insertions and deletions, that holds even for directed acyclic graphs.

\begin{figure*}
\begin{center}
\includegraphics[width=\textwidth, trim = 0mm 0mm 0mm 80mm, clip]{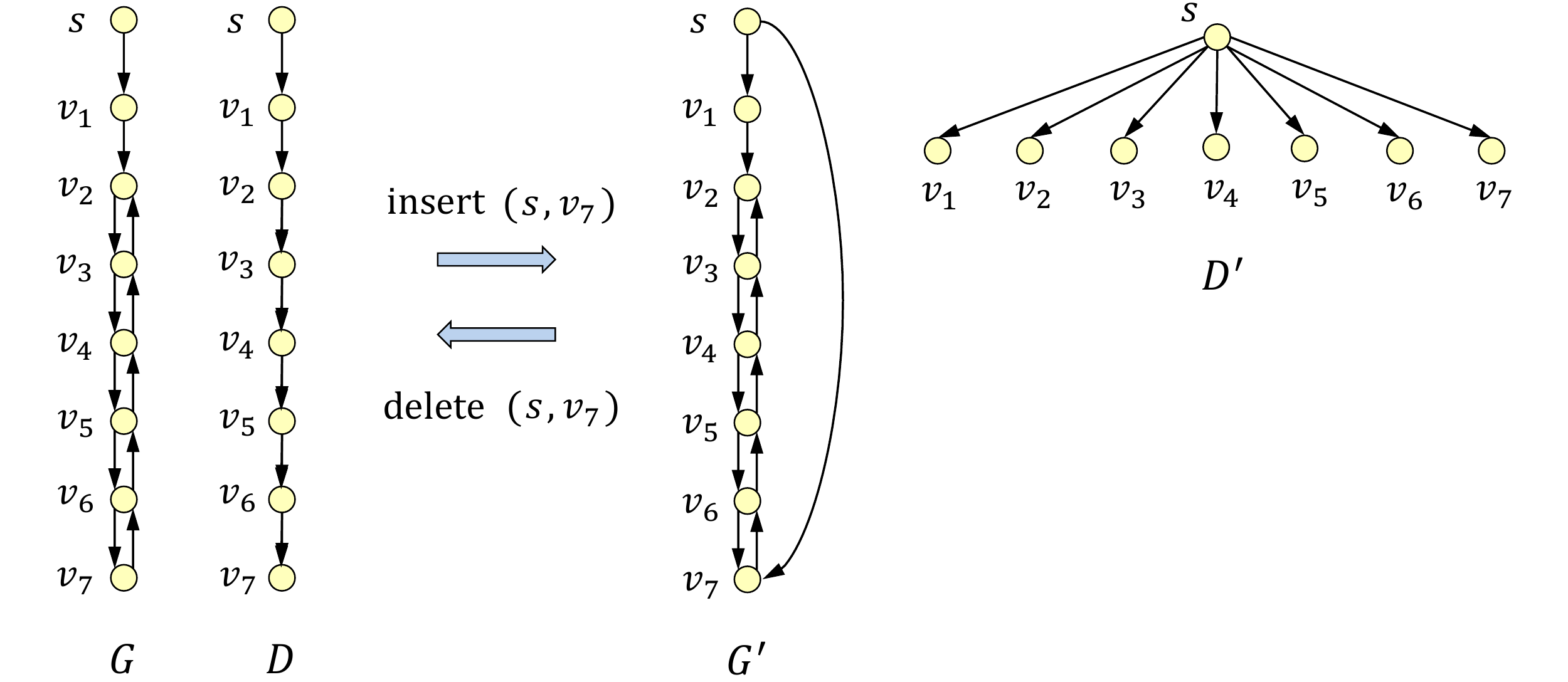}
\caption{\label{fig:pathological1} Pathological updates: Each update (insertion or deletion) affects $n-2$ vertices. (In this instance $n=8$.)}
\end{center}
\end{figure*}

\begin{figure*}
\begin{center}
\includegraphics[width=\textwidth, trim = 0mm 0mm 0mm 80mm, clip]{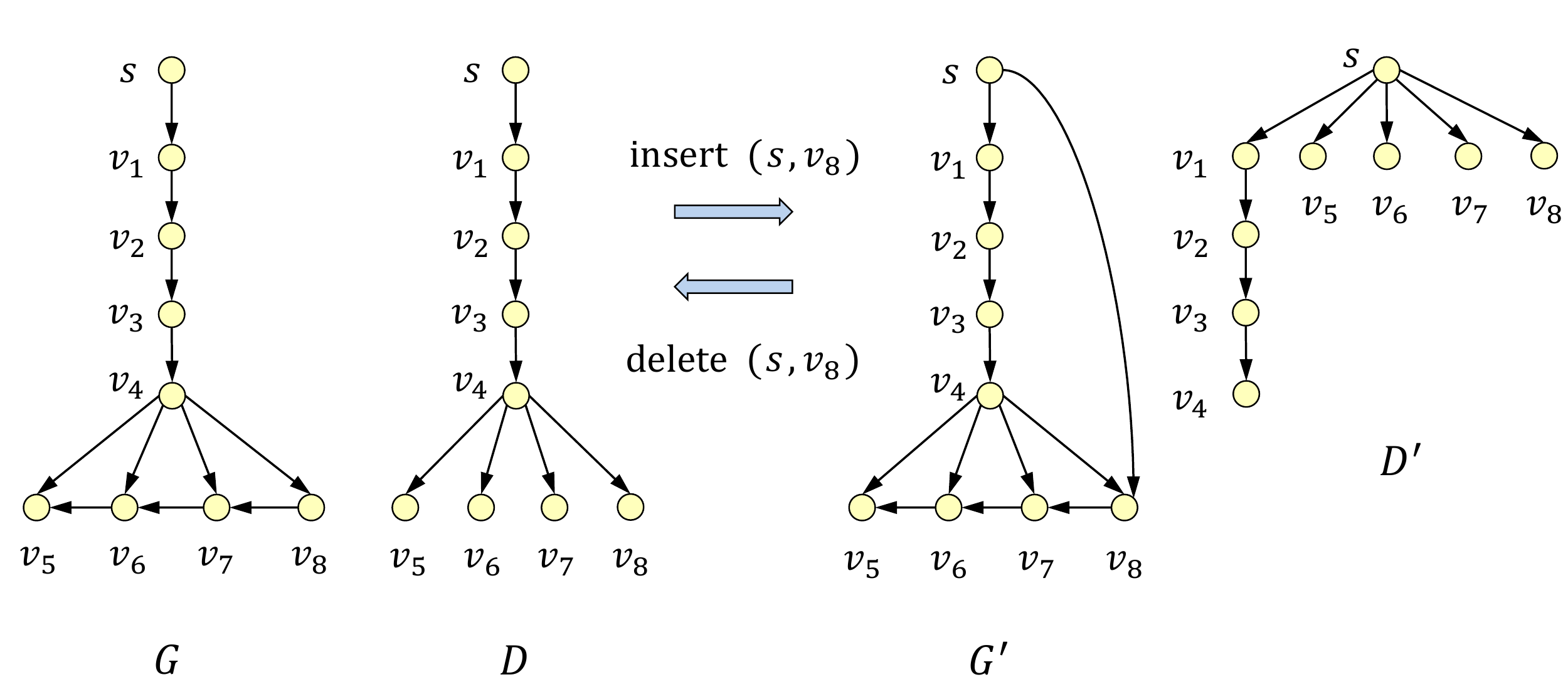}
\caption{\label{fig:pathological2} Pathological updates in the acyclic case: Each update (insertion or deletion) affects $\lceil n/2 \rceil$ vertices. (In this instance $n=9$.)}
\end{center}
\end{figure*}

Using the structural properties of dominator trees stated above we can limit the number of vertices and edges processed during the search for affected vertices.
The following fact is an immediate consequence of the parent and sibling properties of the dominator tree.

\begin{proposition}
\label{proposition:update-properties}
An edge insertion can violate the parent property but not the sibling property of $D$. An edge deletion can violate the sibling property but not the parent property of $D$.
\end{proposition}

Throughout the rest of this paper, $(x,y)$ is the inserted or deleted edge and $x$ is reachable.

\subsection{Edge insertion}
\label{sec:insertion}

We consider two cases, depending on whether $y$ was reachable before the insertion.
Suppose first that $y$ was reachable.
Let $\mathit{nca}_D(x,y)$ be the nearest (lowest) common ancestor of $x$ and $y$ in $D$. If either $\mathit{nca}_D(x,y)=d(y)$ or $\mathit{nca}_D(x,y)=y$ then, by Theorem~\ref{thm:wide-narrow-dom}, the inserted edge has no effect on $D$. Otherwise, the parent property of $D$ implies that $\mathit{nca}_D(x,y)$ is a proper dominator of $d(y)$.
In the following, we denote by $\mathit{depth}(w)$ the depth of vertex $w$ in $D$.

\begin{lemma}
\label{lemma:nca}(\cite{RR:incdom})
Suppose $x$ and $y$ are reachable vertices in $G$. Let $v$ be a vertex that is affected after the insertion of $(x,y)$.
Then $d'(v) = \mathit{nca}_D(x,y)$ and $d'(v)$ is a proper ancestor of $d(v)$ in $D$.
\end{lemma}

For Lemma \ref{lemma:nca} we have that all affected vertices $v$ satisfy $\mathit{depth}(\mathit{nca}_D(x,y)) < \mathit{depth}(d(v)) < \mathit{depth}(v) \le \mathit{depth}(y)$.
Based on the above observations we obtain the following lemma, which is a refinement of a result in~\cite{dynamicdominator:AL}. (See Figure \ref{fig:insert}.)

\begin{lemma}
\label{lemma:insert-affected}
Suppose $x$ and $y$ are reachable vertices in $G$. 
Then, a vertex $v$ is affected after the insertion of the edge $(x,y)$ if and only if $\mathit{depth}(\mathit{nca}_D(x,y)) < \mathit{depth}(d(v))$ and there is a path $P$ from $y$ to $v$ such that $\mathit{depth}(d(v)) < \mathit{depth}(w)$ for all $w \in P$. 
\end{lemma}
\begin{proof}
Suppose $v$ is affected. Let $z = \mathit{nca}_D(x,y)$. By Lemma \ref{lemma:nca} we have
$d'(v) = z$ and that $z$ is an ancestor of $d(v)$ in $D$.
Thus $\mathit{depth}(d'(v)) < \mathit{depth}(d(v))$, and there is a path $P$ from $y$ to $v$ in $G$
that does not contain $d(v)$.
Suppose, for contradiction, that $P$ contains some vertex $w \not= d(v)$ such that $\mathit{depth}(w) \le \mathit{depth}(d(v))$.
Let $P'$ be the part of $P$ from $w$ to $v$. Then, $d(v) \not\in P$ since $P$ does not contain $d(v)$.
The fact that $w \not= d(v)$ and $\mathit{depth}(w) \le \mathit{depth}(d(v))$ implies that $d(v)$ is not an ancestor of $w$ in $D$.
Then, there is a path $Q$ in $G$ from $s$ to $w$ that avoids $d(v)$.
So, the catenation of $Q$ and $P'$ is a path from $s$ to $v$ in $G$ that avoids $d(v)$.
This implies that $d(v)$ does not dominate $v$ before the insertion of $(x,y)$, a contradiction.

To prove the converse, consider a vertex $v$ with $\mathit{depth}(\mathit{nca}_D(x,y)) < \mathit{depth}(d(v))$.
Suppose $G$ contains a path $P$ from $y$ to $v$ such that, for all $w \in P$, $\mathit{depth}(d(v)) < \mathit{depth}(w)$.
We argue that $v$ is affected. First we note that $d(v) \not\in P$, since all vertices on $P$ have larger depth.
%
Also, the fact that $\mathit{depth}(\mathit{nca}_D(x,y)) < \mathit{depth}(d(v))$ implies that $x$ is not a descendant of $d(v)$ in $D$.
Hence, $G$ contains a path $Q$ from $s$ to $x$ that avoids $d(v)$. Then $Q \cdot (x,y) \cdot P$ is a path in $G'$ from
$s$ to $v$ that avoids $d(v)$. Thus $v$ is affected.
\end{proof}

\begin{figure*}
\begin{center}
\includegraphics[width=0.8\textwidth, trim = 0mm 0mm 0mm 80mm, clip]{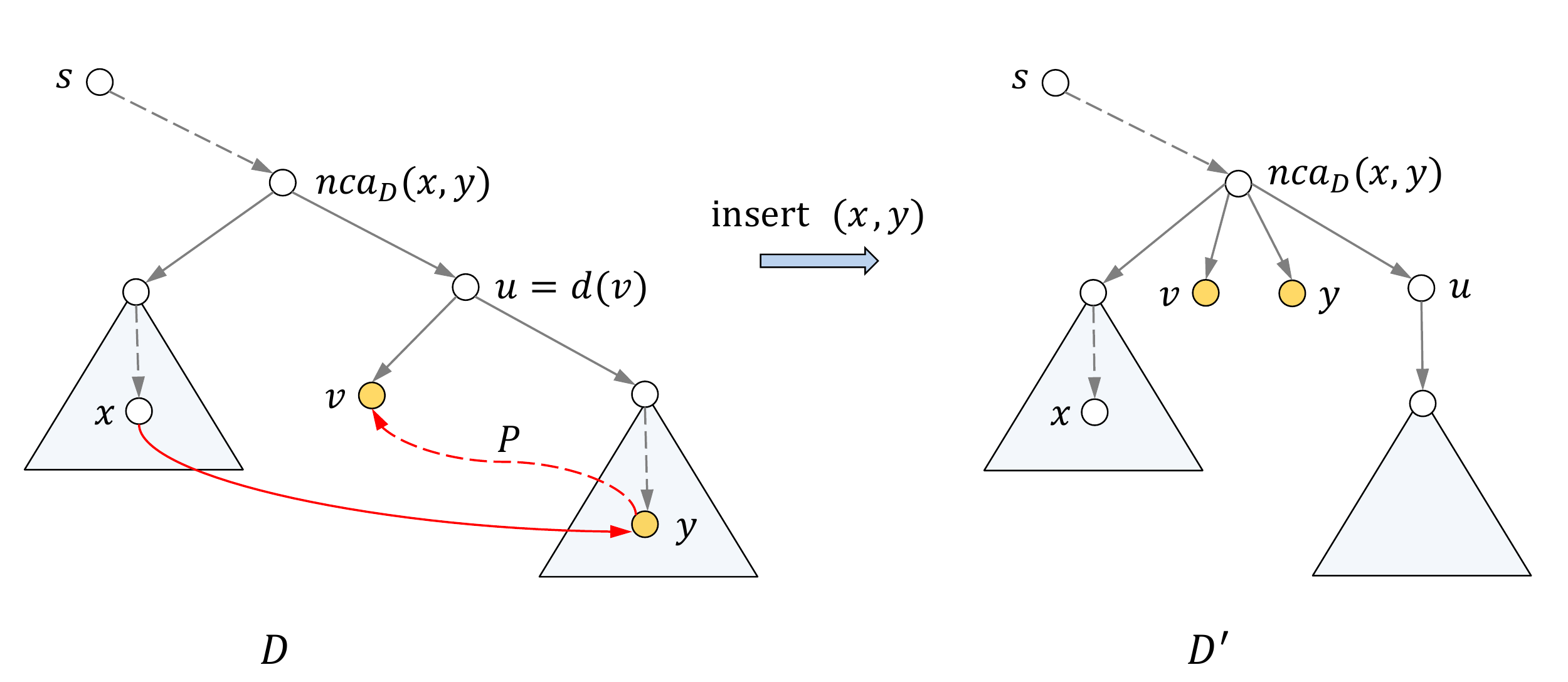}
\caption{\label{fig:insert} Illustration of Lemma \ref{lemma:insert-affected}.}
\end{center}
\end{figure*}

Now suppose $y$ was unreachable before the insertion of $(x,y)$. Then we have $x=d'(y)$. Next, we need to process all other vertices that became reachable after the insertion. To that end, we have three main options:
\begin{itemize}
\item[(a)] Process each edge leaving $y$ as a new insertion, and continue this way until all edges adjacent to newly reachable vertices are processed.
\item[(b)] Compute the set $R(y)$ of the vertices that are reachable from $y$ and were not reachable from $s$ before the insertion of $(x,y)$. We can build the dominator tree $D(y)$ for the subgraph induced by $R(y)$, rooted at $y$, using any static algorithm. After doing that we link $D(y)$ to $x$ by adding the edge $(x,y)$ into $D$. Finally we process every edge $(u,v)$ where $u \in R(y)$ and $v \not\in R(y)$ as a new insertion. Note that the effect on $D$ of each such edge is equivalent to adding the edge $(x,v)$ instead.
\item[(c)] Compute $D'$ from scratch.
\end{itemize}

\subsection{Edge deletion}
\label{sec:deletion}

We consider two cases, depending on whether $y$ becomes unreachable after the deletion of $(x,y)$. Suppose first that $y$ remains reachable. Deletion is harder than insertion because each of the affected vertices may have a different new immediate dominator.
Also, unlike the insertion case, we do not have a simple test, as the one stated in Lemma~\ref{lemma:insert-affected}, to decide whether the deletion affects any vertex. Consider, for example, the following (necessary but not sufficient) condition: 
if $(d(y),y)$ is not an edge of $G$ then there are edges $(u,y)$ and $(w,y)$ such that $\mathit{sp}_D(u,y) \neq \mathit{sp}_D(w,y)$. Unfortunately, this condition may still hold in $D$ after the deletion even when $D' \not= D$. See Figure \ref{fig:deletiontest}.
Despite this difficulty, in Section \ref{sec:dynsnca} we give a simple but conservative test (i.e., it allows false positives) to decide if there are any affected vertices. Also, as observed in~\cite{SGL}, we can limit the search for affected vertices and their new immediate dominators as follows. Since the edge deletion may violate the sibling property of $D$ but not the parent property, it follows that the new immediate dominator of an affected vertex $v$ is a descendant of some sibling of $v$ in $D$. This implies the following lemma that provides a necessary (but not sufficient) condition for a vertex $v$ to be affected.

\begin{figure*}
\begin{center}
\includegraphics[width=\textwidth, trim = 0mm 0mm 0mm 110mm, clip]{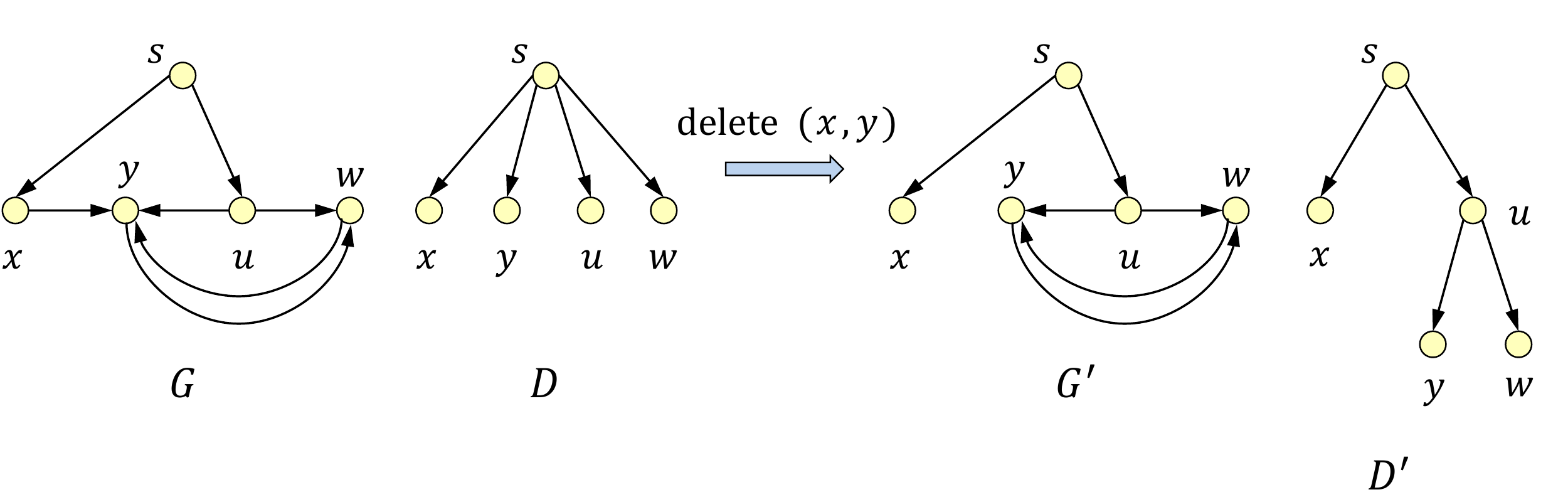}
\caption{\label{fig:deletiontest}
After the deletion of $(x,y)$, $y$ still has two entering edges $(u,y)$ and $(w,y)$ such that $\mathit{sp}_D(u,y) \neq \mathit{sp}_D(w,y)$.}
\end{center}
\end{figure*}

\begin{lemma}
\label{lemma:deletion-reachable}
Suppose $x$ is reachable and $y$ does not becomes unreachable after the deletion of $(x,y)$.
A vertex $v$ is affected only if $d(v)=d(y)$ and there is a path $P$ from $y$ to $v$ such that $\mathit{depth}(d(v)) < \mathit{depth}(w)$ for all $w \in P$.
\end{lemma}
\begin{proof}
Let $G$ be the graph immediately before the deletion, and let $G'$ be the graph immediately after the deletion. Consider the reverse operation, i.e., adding $(x,y)$ to $G'$ to produce $G$. Then $v$ is affected by this insertion, thus by Lemma~\ref{lemma:insert-affected} there is a path $P$ from $y$ to $v$ in $G'$ such that $\mathit{depth}'(w) > \mathit{depth}'(d'(v))$ for all $w \in P$, where $d'(u)$ is the immediate dominator of a vertex $u$ in $G'$,
and $\mathit{depth}'(u)$ is the depth of $u$ in the dominator tree of $G'$. Then, in $G$ (the flow graph that results from $G'$ after the insertion of $(x,y)$),
we have $\mathit{depth}(w) \ge \mathit{depth}(v) > \mathit{depth}(d(v))$, for all $w \in P$.
\end{proof}

Next we examine the case where $y$ becomes unreachable after $(x,y)$ is deleted.
This happens if and only if $y$ has no entering edge $(z,y)$ in $G'$ such that $\mathit{sp}_D(z,y) \not= y$. In this case, all descendants of $y$ in $D$ also become unreachable. The deletion of $y$ and of its descendants in $D$, in turn, may affect other vertices.
The vertices that are possibly affected can be identified by the following lemma.

\begin{lemma}
\label{lemma:deletion-unreachable}
Suppose $x$ is reachable and $y$ becomes unreachable after the deletion of $(x,y)$.
A vertex $v$ is affected only if there is a path $P$ from $y$ to $v$ such that $\mathit{depth}(d(v)) < \mathit{depth}(w)$ for all $w \in P$.
\end{lemma}
\begin{proof}
Let $E^{+}(y) = \{ (u,v) \in E \ | \ y \in \mathit{Dom}(u) \mbox{ and } y \not\in \mathit{Dom}(v)\}$. Consider what happens when we delete the edges in $E^{+}(y)$ one by one (in any order) before deleting $(x,y)$. Let $v$ be a vertex that is affected by the deletion of $(x,y)$ such that $y \not \in \mathit{Dom}(v)$. Then there is a subsequence of deletions of edges $(x_1,y_1),(x_2,y_2),\ldots,(x_k,y_k)$ in $E^{+}(y)$ such that each deletion affects $v$. Note that every deletion $(x_i,y_i)$ in this subsequence leaves $y_i$ reachable. Therefore Lemma~\ref{lemma:deletion-reachable} applies. Each deletion increases the depths of the affected vertices and their descendants, so the result follows.
\end{proof}

As in the case of an insertion that makes a new reachable vertex, when a deletion makes a new unreachable vertex, we can consider the following options:
\begin{itemize}
\item[(a)] Collect all edges $(u,v)$ such that $u$ is a descendant of $y$ in $D$ ($y \in \mathit{Dom}(u)$) but $v$ is not, and process $(u,v)$ as a new deletion. (Equivalently we can substitute $(u,v)$ with $(y,v)$ and process $(y,v)$ as a new deletion.)
\item[(b)] Use a static algorithm to compute the immediate dominators of all possibly affected vertices identified by Lemma \ref{lemma:deletion-unreachable}.
\item[(c)] Compute $D'$ from scratch.
\end{itemize}

\subsection{Reducible flow graphs}
\label{sec:reducible}

A flow graph $G=(V,E,r)$ is \emph{reducible} if every strongly connected subgraph $S$ has a single \emph{entry} vertex $v$ such every path from $s$ to a vertex in $S$ contains $v$~\cite{rdflow:hu74,reducibility:jcss:tarjan}. Tarjan~\cite{reducibility:jcss:tarjan} gave a characterization of reducible flow graphs using dominators: A flow graph is reducible if and only if it becomes acyclic when every edge $(v, w)$ such that $w$ dominates $v$ is deleted. The notion of reducibility is important because many programs have control flow graphs that are reducible, which simplifies many computations. In our context, we can use a dynamic dominators algorithm to dynamically test flow graph reducibility.

\section{Algorithms}
\label{sec:algorithms}

Here we present new algorithms for the dynamic dominators problem. We begin in Section \ref{sec:dynsnca} with a simple dynamic version of the \algo{snca} algorithm (\algo{dsnca}), which also provides a necessary (but not sufficient) condition for an edge deletion to affect the dominator tree. Then, in Section \ref{sec:depth-based}, we present a depth-based search (\algo{dbs}) algorithm which uses the results of Section \ref{sec:properties}. We improve the efficiency of deletions in \algo{dbs} by employing a test for affected vertices used in \algo{dsnca}. In Section \ref{sec:SGL} we give an overview of the Sreedhar-Gao-Lee algorithm~\cite{SGL}.
In the description below we let $(x,y)$ be the inserted or deleted edge and assume that $x$ is reachable.

\subsection{Dynamic SNCA (DSNCA)}
\label{sec:dynsnca}

We develop a simple method to make the (static) \algo{snca} algorithm~\cite{dom_exp:gtw} dynamic, in the sense that it can respond to an edge update faster (by some constant factor) than recomputing the dominator tree from scratch. Furthermore, by storing some side information we can test if the deletion of an edge satisfies a necessary condition for affecting $D$.

The \algo{snca} algorithm
is a hybrid of the simple version of Lengauer-Tarjan (\algo{slt}) and the iterative algorithm of Cooper et al.
The Lengauer-Tarjan algorithm uses the concept of \emph{semidominators}, as an initial approximation to the immediate dominators.
It starts with a depth-first
search on $G$ from $s$ and assigns preorder numbers to
the vertices. Let $T$ be the corresponding depth-first search tree, and let $\mathit{pre}(v)$ be the preorder number of $v$. A path
$P=(u=v_0,v_1,\ldots,v_{k-1},v_k=v)$ in $G$ is a
\emph{semidominator path} if $\mathit{pre}(v_i)>\mathit{pre}(v)$
for $1 \le i \le k-1$. The semidominator of $v$, $\mathit{sd}(v)$, is defined
as the vertex $u$ with minimum $\mathit{pre}(u)$ such that there is a semidominator path from $u$ to
$v$. Semidominators and
immediate dominators are computed
by executing path-minima computations, which find minimum $\mathit{sd}$ values on paths
of $T$, using an appropriate data-structure. Vertices are processed in reverse preorder, which ensures that
all the necessary values are available when needed.
With a simple implementation of the path-minima data structure,
the algorithm \algo{slt} runs in $O(m\log{n})$ time.
With a more sophisticated strategy the algorithm
runs in $O(m\alpha(m,n))$ time.

The \algo{snca} algorithm computes dominators in two phases:
\begin{itemize}
\item[(a)] Compute $\mathit{sd}(v)$ for all $v \neq s$, as done
by {\algo{slt}}.
\item[(b)] Build $D$ incrementally as follows: Process the
vertices in preorder. For each vertex $w$, ascend the path from $t(w)$ to $s$ in $D$,
where $t(w)$ is the parent of $w$ in $T$ (the depth-first search tree), until
reaching the deepest vertex $x$ such that $\mathit{pre}(x) \le
\mathit{pre}(\mathit{sd}(w))$. Set $x$ to be the parent of $w$ in $D$.
\end{itemize}
With a na\"{\i}ve implementation, the second phase runs in $O(n^2)$ worst-case time.
However, as reported in~\cite{strong-articulation:algorithmica,dom_exp:gtw}, 
it performs much better in practice.
{\algo{snca}} is simpler than {\algo{slt}} in that it requires fewer arrays, eliminates some indirect
addressing, and there is one
fewer pass over the vertices. This makes it easier to produce a dynamic version of \algo{snca}, as described below. We note, however, that the same ideas can be applied to produce a dynamic version of \algo{slt} as well.

\paragraph{Edge insertion.}
Let $T$ be the depth-first search tree used to compute semidominators.
Let $\mathit{pre}(v)$ be the preorder number of $v$ in $T$ and $\mathit{post}(v)$ be the postorder number of $v$ in $T$; if $v \not\in T$ then $\mathit{pre}(v) (= \mathit{post}(v)) = 0$. The algorithm runs from scratch if
\begin{equation}
\label{eq:condition1}
\mathit{pre}(y)=0 \mbox{ or,  } \mathit{pre}(x)<\mathit{pre}(y) \ \mbox{ and } \ \mathit{post}(x)<\mathit{post}(y).
\end{equation}
If condition \eqref{eq:condition1} does not hold then $T$ remains a valid depth-first seach tree for $G$. If this is indeed the case then we can repeat the computation of semidominators for the vertices $v$ such that $\mathit{pre}(v) \le \mathit{pre}(y)$. To that end, for each such $v$, we initialize the value of $\mathit{sd}(v)$ to $t(v)$ and perform the path-minima computations for the vertices $v$ with $\mathit{pre}(v) \in \{2,\ldots,\mathit{pre}(y)\}$. Finally we perform the nearest common ancestor phase for all vertices $v \not= r$.

\paragraph{Edge deletion.}

In order to test if the deletion may possibly affect the dominator tree we use the following idea. For any $v \in V-s$, we define $g(v)$ to be a predecessor of $v$
that belongs to a semidominator path from $\mathit{sd}(v)$ to $v$. Such vertices can be
found easily during the computation of semi-dominators~\cite{DomCert:TALG}.

\begin{lemma}
\label{lemma:deletion-test}
The deletion of $(x,y)$ affects $D$ only if $x=t(y)$ or $x=g(y)$.
\end{lemma}

If $x=t(y)$ we run the whole \algo{snca} algorithm from scratch. Otherwise, if $x=g(y)$, then we perform the path-evaluation phase for the vertices $v$ such that $\mathit{pre}(v) \in \{2,\ldots,\mathit{pre}(y)\}$. Finally we perform the nearest common ancestor phase for all vertices $v \not= s$. We note that an insertion or deletion takes $\Omega(n)$ time, since our algorithm needs to reset some arrays of size $\Theta(n)$. Still, as the experimental results given in Section \ref{sec:experimental} show, the algorithm offers significant speedup compared to running the static algorithm from scratch.

\subsection{Depth-based search (DBS)}
\label{sec:depth-based}

This algorithm uses the results of Section \ref{sec:properties} and ideas from~\cite{dynamicdominator:AL,SGL} and \algo{dsnca}. Our goal is to search for affected vertices using the depth of the vertices in the dominator tree, and improve batch processing in the unreachable cases (when $y$ is unreachable before the insertion or $y$ becomes unreachable after a deletion).

\paragraph{Edge insertion.}

In order to locate the vertices that are affected by the insertion of edge $(x,y)$, we start a search in $G$ from $y$ and follow paths that satisfy Lemma~\ref{lemma:insert-affected}.
We say that a vertex $v$ is \emph{scanned}, if the edges leaving $v$ are examined during the search for affected vertices.
Also, we say that $v$ is \emph{visited} if there is a scanned vertex $u$ such
that $(u,v)$ is an edge in $G$ that was examined while scanning $u$.
With each vertex $v$, we store two bits to indicate if $v$ was found to be affected and if $v$ was scanned.

As in \cite{SGL}, we maintain a set of affected vertices, sorted by their depth in $D$.
To do this efficiently, we maintain an array $A$ of $n$ buckets, where bucket $A[i]$ stores the affected vertices $v$ with $\mathit{depth}(v)=i$.
When we find a new affected vertex $v$ we insert it into $A[\mathit{depth}(v)]$.
We also maintain the most recently scanned affected vertex $\widehat{v}$, and the affected level $\widehat{\ell}=\mathit{depth}((\widehat{v}))$.

Initially, all vertices are marked as not affected and not scanned. Also, all buckets $A[i]$ are empty.
When $(x,y)$ is inserted, we locate $z=\mathit{nca}_D(x,y)$ and test if $\mathit{depth}(z) < \mathit{depth}(d(y))$.
If this is the case, then we mark $y$ as affected and insert it into bucket $A[\mathit{depth}(y)]$.
While there is a non-empty bucket, we locate the largest index $\ell$ such that $A[\ell]$ is not empty, and
extract a vertex $v$ from $A[\ell]$. Then, we set $\widehat{v} = v$ and $\widehat{\ell}=\mathit{depth}(v)$, and scan $v$.
To scan a vertex $v$, we examine the edges $(v,w)$ that leave $v$.
Let $(v,w)$ be the an edge that we examine.
If $\mathit{depth}(w) > \widehat{\ell}$ then we recursively scan $w$ if it was not scanned before.
If $\mathit{depth}(\mathit{nca}_D(x,y)) + 1 < \mathit{depth}(w) \le \widehat{\ell}$ then we mark $w$ as affected and insert $w$ into bucket
$A[\mathit{depth}(w)]$.

\begin{lemma}
\label{lemma:invariants}
During the insertion of an edge $(x,y)$, where $y$ is affected, algorithm \algo{dbs} maintains the following invariants:
\begin{itemize}
\item[(1)] The affected level $\widehat{\ell}$ is non-increasing, and $\widehat{\ell} > \mathit{depth}(\mathit{nca}_D(x,y))+1$.
\item[(2)] For any scanned vertex $v$, $\mathit{depth}(v) \ge \widehat{\ell}$.
\item[(3)] Suppose $v$ is scanned when the affected level is $\widehat{\ell}$. Then, there is a path from $y$ to $v$ that contains only vertices of depth $\widehat{\ell}$ or higher.
\item[(4)] If there is a path from $y$ to $v$ that contains vertices of minimum depth $\ell > \mathit{depth}(\mathit{nca}_D(x,y)) + 1$, then $v$ is scanned when the affected level is $\widehat{\ell} \ge \ell$.
\item[(5)] Any vertex is scanned at most once, and a scanned vertex is a descendant in $D$ of an affected vertex.
\end{itemize}
\end{lemma}
\begin{proof}
Invariants (1), (2), and (3) follow immediately from the description of the algorithm.
Also, invariants (5) is implied by invariant (4). So it suffices to prove
that the algorithm maintains invariant (4).

Let $v$ be a vertex such that there is a path $P$ from $y$ to $v$ with
minimum vertex depth $\ell > \mathit{depth}(\mathit{nca}_D(x,y))+1$. Assume, for contradiction, that $v$ is not scanned
for $\widehat{\ell} \ge \ell$.
Choose $v$ so that the length of $P$ is minimum. Let $u$ be the vertex that precedes $v$ on $P$.
Since $y$ is scanned, $v \not= y$ and so vertex $u$ exists.
Then, by the choice of $v$, we have that $u$ is scanned.
So the edge $(u,v)$ is examined, and since $\mathit{depth}(\mathit{nca}_D(x,y))+1 < \ell \le \mathit{depth}(v)$,
$v$ will be scanned when $\widehat{\ell} \ge \ell$, a contradiction.
\end{proof}


The correctness of our algorithm follows by invariant (5), which implies that all affected vertices will be detected.

Now suppose that $y$ was unreachable before the insertion. Then,
we can apply one of the approaches mentioned in Section \ref{sec:insertion}.
In order to provide a good worst-case bound for
a sequence of $k$ edge insertions, we will assume that we compute $D'$ from scratch.

\begin{theorem}
\label{theorem:dbs-insertions}
Algorithm \algo{dbs} maintains the dominator tree of a flow graph through a sequence of $k$ edge insertions in $O(m \min\{k,n\} + kn)$ time, where $n$ is the number of vertices and $m$ is the number of edges after all insertions.
\end{theorem}
\begin{proof}
Let $(x,y)$ be an edge that is inserted into $G$ during the insertion sequence.
We test can if $x$ and $y$ are reachable from $s$ before the insertion in $O(1)$ time, and then consider the following cases:
\begin{itemize}
\item[(a)] $x$ is unreachable. We only need to update the adjacency lists of $G$, which takes $O(1)$ time.

\item[(b)] $x$ is reachable and $y$ is unreachable. We compute the dominator tree from scratch in $O(m)$ time. Such an event can happen at most $\min\{k,n\}$ times throughout the sequence, so the total time spent on these type of insertions is $O(m \min\{k,n\})$.

\item[(c)] $x$ and $y$ are reachable. We first compute $\mathit{nca}_D(x,y)$ in $O(n)$ time, just by following parent pointers in $D$.
If $y$ is affected, then we execute the depth-based search algorithm to locate the affected vertices. Let $\nu$ be the number of scanned vertices, and let
$\mu$ be the total number of edges leaving a scanned vertex. Excluding the time to needed to maintain the buckets $A[i]$, the search for affected vertices takes $O(\nu+\mu)$ time. So we can charge a cost of $O(1+\mathit{outdeg}(v))$ to each scanned vertex $v$. To bound the total time for all insertions of this type we note that by invariant (5) of Lemma \ref{lemma:invariants}, each time a vertex is scanned its depth in the dominator tree will decrease by at least one. Also, by the same invariant we have that each vertex will be scanned at most $\min\{k,n\}$ times. Hence, the total cost per vertex is
$\mathit{cost}(v) = O((1+\mathit{outdeg}(v))\min\{k,n\})$.
Finally, we consider the time required to maintain the buckets $A[i]$. Each affected vertex $v$ is inserted into and deleted from a single bucket $A[i]$,
and each such operation takes constant time. It remains to bound the time required to locate the nonempty buckets. Invariant (1) of Lemma \ref{lemma:invariants} implies that we need to test if $A[i]$ is not null only once for each depth $i$. Hence, we can maintain all buckets in $O(n)$ time.
We conclude that the total time spent on all insertions of type (c) is bounded by $O(kn) + \sum_v{\mathit{cost}(v)} =
O(m \min\{k,n\} +kn)$.
\end{itemize}
Hence we get a total $O(m \min\{k,n\} + kn)$ bound for all $k$ insertions.
\end{proof}

It is straightforward to extend our algorithm so that it can answer in constant time the following type of queries:
Given two vertices $u$ and $v$, test if $u$ dominates $v$ in $G$.
We can do this test using an $O(1)$-time test of the ancestor-descendant relation in $D$~\cite{domin:tarjan}.
E.g., we can number the vertices of $D$ from $1$ to $n$ in preorder and compute the number of descendants of each vertex $w$;
we denote these numbers by $\mathit{pre}(w)$ and $\mathit{size}(w)$, respectively.
Then $v$ is a descendant of $u$ if and only if $\mathit{pre}(u) \le \mathit{pre}(v) < \mathit{pre}(u) + \mathit{size}(u)$.
We can recompute these numbers in $O(n)$ time after each insertion, so the bound of Theorem \ref{theorem:dbs-insertions}
is maintained.

\paragraph{Edge deletion.}

We describe a method that applies \algo{snca}.
If $y$ is still reachable after the deletion then we execute \algo{snca} for the subgraph induced by $d(y)$. Now suppose $y$ becomes unreachable. Let $E^{+}(y) = \{ (u,v) \in E \ | \ y \in \mathit{Dom}(u) \mbox{ and } y \not\in \mathit{Dom}(v)\}$ and let $V^{+}(y) = \{ v \in V \ | \ \mbox{there is an edge } (u,v) \in E^{+}(y) \}$. We compute $V^{+}(y)$ by executing a depth-first search from $y$, visiting only vertices $w$ with $\mathit{depth}(w) \ge \mathit{depth}(y)$. At each visited vertex $w$ we examine the edges $(w,v)$ leaving $w$; $v$ is included in $V^{+}(y)$ if $\mathit{depth}(v) \le \mathit{depth}(y)$. Next, we find a vertex $v \in V^{+}(y)$ of minimum depth such that $v \not \in \mathit{Dom}(y)$. Finally we execute \algo{snca} for the subgraph induced by $d(v)$. Lemma~\ref{lemma:deletion-unreachable} implies the correctness of this method. A benefit of this approach is that we can apply Lemma~\ref{lemma:deletion-test} to test if the deletion may affect $D$. Here we have the additional complication that we can maintain the $t(y)$ and $g(y)$ values required for the test only for the vertices $y$ such that $d(y)$ was last computed by \algo{snca}. Therefore, when an insertion affects a reachable vertex $y$ we set $t(y)$ and $g(y)$ to $\mathit{null}$ and cannot apply the deletion test if some edge entering $y$ is deleted.

\subsection{Sreedhar-Gao-Lee algorithm}
\label{sec:SGL}

This algorithm uses the DJ-graph structure~\cite{SG:phinodes} to allow fast search of the affected vertices.
The DJ-graph maintains two sets of edges, $E_d$  which stores the dominator tree edges ($d$-edges), and $E_j$ which stores the edges in $E \setminus E_d$, called \emph{join-edges} ($j$-edges). The two sets are stored in different adjacency lists to allow fast search of the vertices affected by the update operations. The affected vertices are a subset of the \emph{iterated dominance frontier}~\cite{IDF} of $y$, denoted as $\mathit{IDF(y)}$. The \emph{dominance frontier}, $\mathit{DF(y)}$, of a vertex $y$ is the set of vertices $z$ such that $y$ dominates a predecessor of $z$ but does not properly dominate $y$. For a set of vertices $S \subseteq V$ we define $\mathit{DF(S)}=\bigcup_{z \in S}\mathit{DF(z)}$. Then, $\mathit{IDF(S)}$ is the limit of $\mathit{IDF}_{i}(S)$, defined by the recursion $\mathit{IDF}_{i+1}(S)=\mathit{DF}(S \cup \mathit{IDF}_{i}(S))$ and $\mathit{IDF}_{1}(S) = \mathit{DF}(S)$. The algorithm also maintains the depth of each vertex in $D$.

\paragraph{Edge insertion.}
As shown in~\cite{SGL} the affected vertices are exactly those in $\mathit{IDF(y)}$ that satisfy $\mathit{depth}(z) > \mathit{depth}(\mathit{nca}_D(x,y))+1$. To compute this set, the algorithm maintains a set of affected vertices $A$ sorted by their depth in $D$, using buckets. (We used the same idea in \algo{dbs}.)
Initially $A=\{y\}$, and while $A$ is not empty, a vertex $v \in A$ with maximum depth is extracted and processed. To process $v$, the algorithm visits the subtree of $D$ rooted at $v$ (using the $d$-edges) and at each visited vertex $w$ it looks at the leaving $j$-edges $(w,u)$. If $\mathit{depth}(v) \ge \mathit{depth}(u) > \mathit{depth}(\mathit{nca}_D(x,y))+1$ then $u$ is affected and is inserted into $A$. The inequality $\mathit{depth}(u) \le \mathit{depth}(v)$ makes it easy to extract in amortized constant time a vertex in $A$ with maximum depth to be processed next. The case where $y$ was unreachable is handled as described in Section~\ref{sec:properties}. That is, the algorithm computes the dominator tree $D(y)$ induced by $R(y)$, using the algorithm of Cooper et al. Then it adds the edge $(x,y)$ in $D$, and processes each edge $(x',y')$ with $x' \in R(y)$ and $y' \not\in R(y)$ as a new insertion.

\paragraph{Edge deletion.}
The deletion needs to consider the vertices of $v \in \mathit{IDF}(y)$ with  $\mathit{depth}(v) = \mathit{depth}(y)$. These are siblings of $y$ in $D$ but not all of them are necessarily affected. The method to handle deletions suggested in~\cite{SGL} applies the searching procedure as in the insertion case to identify a set of possibly affected vertices. Then it uses the
set-intersecting iterative algorithm for the possibly affected vertices to compute their dominators.\footnote{In~\cite{SGL} this algorithm is incorrectly cited as the Purdom-Moore algorithm~\cite{dom:pm72}.} In our implementation we use a simpler idea that improves the performance of deletions. Namely, we can apply Lemma~\ref{lemma:deletion-reachable} to find the set of possibly affected vertices and then use the algorithm of Cooper et al. to compute the immediate dominators of this set of vertices. The algorithm handles the case where $y$ becomes unreachable after $(x,y)$ is deleted as described in Section~\ref{sec:properties}. Using depth-first search we can collect all edges $(u,v)$ such that $u$ is a descendant of $y$ in $D$ but $v$ is not, and process $(u,v)$ as a new deletion.

\section{Experimental Evaluation}
\label{sec:experimental}

\subsection{Implementation and Experimental Setup}

We evaluate the performance of four algorithms: the simple version of Lengauer-Tarjan (\algo{slt}), dynamic \algo{snca} (\algo{dsnca}), an efficient implementation of Sreedhar-Gao-Lee (\algo{sgl})
and the depth-based search algorithm (\algo{dbs}). In this setting, \algo{slt} runs the simple version of the Lengauer-Tarjan algorithm after each edge $(x,y)$ update, but only if $x$ is currently reachable from $s$. We do not report running times for static \algo{snca} as they are very close to those of \algo{slt}.
We
implemented all algorithms in C++. They take as input the graph
and its root, and maintain an $n$-element array representing
immediate dominators. Vertices are assumed to be integers from 1 to
$n$.
The code was compiled using {\tt g++} v. 3.4.4 with full
optimization (flag {\tt -O4}). All tests were conducted on an
Intel Core i7-920 at 2.67GHz with 8MB cache, running Windows Vista Business Edition.
We report CPU times measured with the
\texttt{getrusage} function. \ignore{Since the precision of \texttt{getrusage} is only 1/60
second, we ran each algorithm repeatedly 100 times;
individual times were obtained by dividing the total time by the
number of runs. Note that this strategy artificially reduces the number
of cache misses for all algorithms, since the graphs are usually small.}
Running times include allocation and deallocation of arrays and linked lists, as
required by each algorithm, but do not include reading the graph from an input file.
Our source code is available upon request.

\begin{table}
\begin{center}
\resizebox{\textwidth}{!}{
\begin{tabular}{rrrrrrrrrrrrrrrr}
\hline
\multicolumn{1}{l}{\tabhead graph} && \multicolumn{2}{r}{\tabhead instance} && \multicolumn{1}{r}{\tabhead insertions} && \multicolumn{1}{r}{\tabhead deletions} && \multicolumn{1}{r}{\tabhead slt} && \multicolumn{1}{r}{\tabhead dsnca} && \multicolumn{1}{r}{\tabhead sgl} && \multicolumn{1}{r}{\tabhead dbs} \\
                            && $i$& $d$ &&     &&     &&       &&        &&        &&        \\
\hline
\multicolumn{1}{l}{\textsf{uloop}}   && 10 &  0  && 315 &&   0 && 0.036 && 0.012  && 0.008  && \textbf{0.006}  \\
\multicolumn{1}{l}{$n=580$} &&  0 & 10  &&   0 && 315 && 0.039 && \textbf{0.012}  && 0.059  && \textbf{0.012}  \\
\multicolumn{1}{l}{$m=3157$}&& 10 & 10  && 315 && 315 && 0.065 && 0.017  && 0.049  && \textbf{0.014}  \\
                            && 50 &  0  && 1578&&   0 && 0.129 && 0.024  && 0.013  && \textbf{0.012}  \\
                            &&  0 & 50  &&    0&& 1578&& 0.105 && 0.024  && 0.121  && \textbf{0.023}  \\
                            && 50 & 50  && 1578&& 1578&& 0.067 && \textbf{0.015}  && 0.041  && 0.033  \\
                            && 100&  0  && 3157&&    0&& 0.178 && 0.033  && 0.023  && \textbf{0.019}  \\
                            && 0  & 100 &&    0&& 3157&& 0.120 && \textbf{0.026}  && 0.136  && 0.050  \\
\hline
\multicolumn{1}{l}{\textsf{baydry}}    && 10 &  0  && 198&&    0 && 0.010&&  \textbf{0.002} &&  0.004  && 0.003 \\
\multicolumn{1}{l}{$n=1789$}  &&  0 & 10  &&   0 &&  198 && 0.014 &&  \textbf{0.003} &&  0.195  && 0.004 \\
\multicolumn{1}{l}{$m=1987$}  && 10 & 10  && 198 &&  198 && 0.020 &&  \textbf{0.003} &&  0.012  && 0.005 \\
                              && 50 &  0  &&  993&&    0 && 0.034 && 0.009 && \textbf{0.007} &&  0.008 \\
                              &&  0 & 50  &&    0&&  993 && 0.051 && \textbf{0.007} && 0.078 &&  0.012 \\
                              && 50 & 50  &&  993&&  993 && 0.056 && \textbf{0.006} && 0.033 &&  0.020 \\
                              && 100&  0  && 1987&&    0 && 0.048 && \textbf{0.004} && 0.013 &&  0.015 \\
                              && 0  & 100 &&    0&& 1987 && 0.064 && \textbf{0.010} && 0.095 &&  0.022 \\
\hline
\multicolumn{1}{l}{\textsf{rome99}}    && 10 &  0  && 887&&    0 && 0.261&&  0.106 && 0.027 &&  \textbf{0.017} \\
\multicolumn{1}{l}{$n=3353$}  &&  0 & 10  &&   0&&  887 && 0.581&&  \textbf{0.252} && 1.861 &&  0.291 \\
\multicolumn{1}{l}{$m=8870$}  && 10 & 10  && 887&&  887 && 0.437&&  0.206 && 0.863 &&  \textbf{0.166} \\
                              && 50 &  0  &&4435&&    0 && 0.272&&  0.106 && 0.049 &&  \textbf{0.031} \\
                              &&  0 & 50  &&   0&& 4435 && 1.564&&  \textbf{0.711} && 4.827 &&  0.713 \\
                              && 50 & 50  && 4435&& 4435 &&0.052&&  \textbf{0.016} && 0.074 &&  0.065 \\
                              && 100&  0  && 8870&&    0 &&0.288&&  0.103 && 0.068 &&  \textbf{0.056} \\
                              && 0  & 100 &&    0&& 8870 &&1.274&&  \textbf{0.613} && 4.050 &&  0.639 \\
\hline
\multicolumn{1}{l}{\textsf{s38584}}    && 10 &  0  && 3449&&    0 && 6.856 &&  2.772 &&  0.114  && \textbf{0.096}  \\
\multicolumn{1}{l}{$n=20719$} &&  0 & 10  &&    0&& 3449 && 7.541 &&  \textbf{4.416} && 15.363  && 4.803  \\
\multicolumn{1}{l}{$m=34498$} && 10 & 10  && 3449&& 3449 && 6.287 &&  3.586 &&  5.131  && \textbf{2.585}  \\
                              && 50 &  0  &&17249&&    0 && 9.667 &&  3.950 &&  0.228  && \textbf{0.150}  \\
                              &&  0 & 50  &&    0&&17249 &&10.223 &&  \textbf{5.671} && 17.543  && 5.835  \\
                              && 50 & 50  &&17249&&17249 && 0.315 &&  \textbf{0.107} &&  0.342  && 0.291  \\
                              && 100&  0  &&34498&&    0 &&10.477 &&  4.212 &&  0.301  && \textbf{0.285}  \\
                              && 0  & 100 &&    0&&34498 &&10.931 &&  6.056 && 18.987  && \textbf{6.016}  \\
\hline
\multicolumn{1}{l}{\textsf{p2p-Gnutella25}}    && 10 &  0  &&5470 &&    0 && 38.031&&  9.295 &&  0.167  &&  \textbf{0.123} \\
\multicolumn{1}{l}{$n=22687$}         &&  0 & 10  &&   0 && 5470 && 38.617&& \textbf{13.878} && 38.788  && 16.364 \\
\multicolumn{1}{l}{$m=54705$}         && 10 & 10  &&5470 && 5470 && 72.029&& 21.787 && 37.396  && \textbf{14.767} \\
                                      && 50 &  0  &&27352&&    0 &&129.668&& 37.206 &&  0.415  &&  \textbf{0.256} \\
                                      &&  0 & 50  &&    0&&27352 &&133.484&& \textbf{49.730} &&131.715  && 51.764 \\
                                      && 50 & 50  &&27352&&27352 && 60.776&& 27.996 && 28.478  && \textbf{19.448} \\
                                      && 100&  0  &&54705&&    0 &&136.229&& 39.955 &&  0.724  &&  \textbf{0.468} \\
                                      && 0  & 100 &&    0&&54705 &&128.738&& 54.405 &&139.449  && \textbf{44.064} \\
\hline
\end{tabular}
}
\vspace{.2cm}
\caption{Average running times in seconds for 10 seeds. The best result in each row is bold.} 
\label{tab:runningtimes} 
\end{center}
\end{table}

\subsection{Instances and Evaluation}
\label{sec:instances}

Our test set consists of a sample of graphs used in~\cite{dom_exp:gtw}, graphs taken from the Stanford Large Network Dataset Collection~\cite{StanfordGraphs}, and road networks~\cite{rome99}. We report running times for a representative subset of the above test set, which consist of the following: the control-flow graph \textsf{uloop} from the SPEC 2000 suite created by the IMPACT compiler, the foodweb \textsf{baydry}, the VLSI circuit \textsf{s38584} from the ISCAS'89 suite, the peer-to-peer network \textsf{p2p-Gnutella25}, and the road network \textsf{rome99}. We constructed a sequence of update operations for each graph by simulating the update operations as follows.
Let $m$ be the total number of edges in the graph. We define parameters $i$ and $d$ which correspond, respectively, to the fraction of edges to be inserted and deleted. This means that $m_i = i*m$ edges are inserted and $m_d = d*m$ edges are deleted, and the flow graph initially has $m'=m-m_i$ edges. The algorithms build (in static mode) the dominator tree for the first $m'$ edges in the original graph file
and then they run in dynamic mode. For $i=d=0$, \algo{sgl} reduces to the iterative algorithm of Cooper et al.~\cite{dom:chk01}, whilst \algo{dsnca} and \algo{dbs} reduce to \algo{snca}. The remaining edges are inserted during the updates. The type of each update operation is chosen uniformly at random, so that there are $m_i$ insertions interspersed with $m_d$ deletions. During this simulation that produces the dynamic graph instance we keep track of the edges currently
present in the graph. If the next operation is a deletion then the edge to be deleted is chosen uniformly at random from the edges in the current graph.

\subsection{Evaluation}
\label{sec:evaluation}

The experimental results for various combinations of $i$ and $d$ are shown in Table~\ref{tab:runningtimes}. The reported running times for a given combination of $i$ and $d$ is the average of the total running time taken to process ten update sequences obtained from different seed initializations of the \texttt{srand} function.
With the exception of \textsf{baydry} with $i=50, d=0$, in all instances \algo{dsnca} and \algo{dbs} are the fastest. In most cases, \algo{dsnca} is by a factor of more than 2 faster than \algo{slt}. \algo{sgl} and \algo{dbs} are much more efficient when there are only insertions ($d=0$), but their performance deteriorates when there are deletions ($d>0$). For the $d>0$ cases, \algo{dbs} and \algo{dsnca} have similar performance for most
instances, which is due to employing the deletion test of Lemma~\ref{lemma:deletion-test}. On the other hand, \algo{sgl} can be even worse than \algo{slt} when $d>0$. For all graphs except \textsf{baydry} (which is extremely sparse) we observe a decrease in the running times for the $i=50,d=50$ case. In this case, many edge updates occur in unreachable parts of the graph. (This effect is more evident in the $i=100,d=100$ case, so we did not include it in Table~\ref{tab:runningtimes}.)
Overall, \algo{dbs} achieves the best performance. \algo{dsnca} is a good choice when there are deletions and is a lot easier to implement.

\clearpage

\bibliographystyle{plain}
\bibliography{ltg}

\end{document}